\documentclass[draftclsnofoot, 12pt]{IEEEtran}

\usepackage{graphicx, epstopdf}
\usepackage{amsmath, amssymb, latexsym}
\usepackage{color}
\usepackage{multirow}
\usepackage{array}
\usepackage{hhline}

\begin{document}

\title{Analysis of the Game-Theoretic Modeling of Backscatter Wireless Sensor Networks under Smart Interference}


\author{Seung Gwan Hong,~\IEEEmembership{Student Member,~IEEE}, Yu Min Hwang,~\IEEEmembership{Member,~IEEE},\\ Sun Yui Lee,~\IEEEmembership{Student Member,~IEEE}, Yoan Shin, Dong In Kim,~\IEEEmembership{Senior Member,~IEEE}, and Jin Young Kim$^\dagger$,~\IEEEmembership{Senior Member,~IEEE}
\thanks{This technical report is a detailed version of a paper which was accepted to appear in IEEE Communications Letters. Copyright belongs to IEEE.}
\thanks{This work was, in part, supported by the National Research Foundation of Korea (NRF) grant funded by the Korean government (MSIT) (2014R1A5A1011478), and in part by “Human Resources Program in Energy Technology” of the Korea Institute of Energy Technology Evaluation and Planning (KETEP) grant funded by the Korean government (MOTIE) (No. 20174010201620)}
\thanks{S. G. Hong, Y. M. Hwang, S. Y. Lee and J. Y. Kim are with the Department of Wireless Communications, Kwangwoon University, Seoul 01897, South Korea (email: mygwan112, yumin, sunyuil22, jinyoung@kw.ac.kr). }
\thanks{Y. Shin is with the School of Electrical Engineering, Soongsil University, Seoul 06978, South Korea (email: yashin@ssu.ac.kr).}
\thanks{D. I. Kim is with the School of Information and Communication
Engineering, Sungkyunkwan University, Suwon 16419, South Korea (email: dikim@skku.ac.kr).}
\thanks{$^\dagger$ Corresponding author}
}

\markboth{[Technical Report]}%
{Shell \MakeLowercase{\textit{et al.}}: Bare Demo of IEEEtran.cls for Journals}

\maketitle

\begin{abstract}
In this paper, we study an interference avoidance scenario in the presence of a smart interferer which can rapidly observe the transmit power of a backscatter wireless sensor network (WSN) and effectively interrupt backscatter signals. We consider a power control with a sub-channel allocation to avoid interference attacks and a time-switching ratio for backscattering and RF energy harvesting in backscatter WSNs. We formulate the problem based on a Stackelberg game theory and compute the optimal transmit power, time-switching ratio, and sub-channel allocation parameter to maximize a utility function against the smart interference. We propose two algorithms for the utility maximization using Lagrangian dual decomposition for the backscatter WSN and the smart interference to prove the existence of the Stackelberg equilibrium. Numerical results show that the proposed algorithms effectively maximize the utility, compared to that of the algorithm based on the Nash game, so as to overcome smart interference in backscatter communications. 
\end{abstract}



\section{Introduction}
Recently backscatter wireless sensor networks (WSNs), where a massive number of distributed sensors simultaneously perform backscatter communications for information transmission and radio frequency (RF) energy harvesting, have been attracting the upsurge of interests \cite{intro_rf, intro_rf2}.

However, backscatter WSNs are very vulnerable to interference because of a very low signal strength in their backscattering signals. Especially, in a smart interfering environment where backscatter signals are detected and attacked by the intentional smart interferer, backscatter sensors are required to establish transmission strategies to overcome or avoid the smart interference. Therefore, the motivation of this paper is to solve the problem of how to overcome the smart interference in backscatter WSNs and obtain the best utility of backscatter WSNs. The various previous works, such as utilizing harvest-then-transmit protocol \cite{WPCN1, WPCN2}, a Nash game theory-based jamming defense with power control \cite{Jamming} and medium access control game \cite{Mac_game}, and a Stackelberg game-based smart jamming problem \cite{Stackelberg_Jamming1} by a smart jammer that throws threats on cognitive radio networks with power control \cite{Stackelberg_Jamming2}, have been studied. However, optimal transmission strategy in the smart interference environment that maximizes its own utility, while spying on each other's transmission strategy in backscatter communications with RF energy harvesting, has not yet been investigated.

In this paper, we formulate the smart interference problem based on the Stackelberg game theory as a hierarchical structure-based leader-follower game. Our game model is that two players (a smart interferer and a backscatter WSN) are alternately optimizing their utilities as they explore their opponents' transmission strategies \cite{Stackelberg}. The smart interferer observes the backscatter WSN’s signal power and calculates optimal interference power to reduce the signal-to-interference-plus-noise ratio (SINR) of the backscatter WSN considering the smart interferer’s own battery energy. Then the backscatter WSN observes the smart interferer’s interference power and shifts the current sub-channel to a different sub-channel with less interference while performing power control. Therefore, we try to solve the smart interference problem with the hierarchical structure-based Stackelberg game theory rather than the simultaneous structure-based Nash game theory.

The contribution of this paper is to propose a new transmission strategy based on Stackelberg game theory and Lagrangian dual decomposition that effectively improves the utility of backscatter WSNs in the smart interference environment. To overcome the smart interference in backscatter WSNs, the proposed transmission strategy is composed of two resource allocation algorithms allocating optimal transmit power, time-switching ratio, and sub-channel allocation parameter under practical constraints. Furthermore, with the proposed algorithms, we prove the existence of the Stackelberg equilibrium which means a converged value that can no longer increase the utility of backscatter WSNs.

\subsection{Operation of Backscatter WSN under Smart Interferer}
As shown in Fig. \ref{fig:model}, backscatter tags can harvest energy from transmitted signals from a hybrid-access point (H-AP) and reflect the signal to transmit its own backscatter data signal back to the H-AP. We consider a backscatter WSN consisting of $n$=\{1,$\cdots$,$N$\} tags, and $k$=\{1,$\cdots$,$K$\} represents sub-channels which are shifted frequencies for backscattering \cite{BC_Shift}. The backscatter tags are equipped with a single antenna and operated by time-division multiple access (TDMA). We consider a transmission protocol operating within one block transmission time $T$ and the time-switching ratio allocated to time for energy harvesting as $(1-\rho)T$, and allocated to time for backscattering as $\rho T$. Hence, the amount of the harvested energy at the $n$-th tag can be represented as:
\begin{equation}\label{eq:E_n}
\begin{split}
E_n=\eta (1-\rho)Th_nP_{t,n}
\end{split},
\end{equation}
where $P_{t,n}$ is the transmit power at the H-AP and $\eta$ denotes the energy harvesting efficiency for all tags. $h_n$ is channel gain from the H-AP to the $n$-th tag, which is defined as: 
\begin{equation}\label{eq:h_n}
\begin{split}
h_n=G_tG_r\lambda^2_B/(4\pi r_n)^2
\end{split}.
\end{equation}
\begin{figure}[t]
\centering
\includegraphics[width=6.5in]{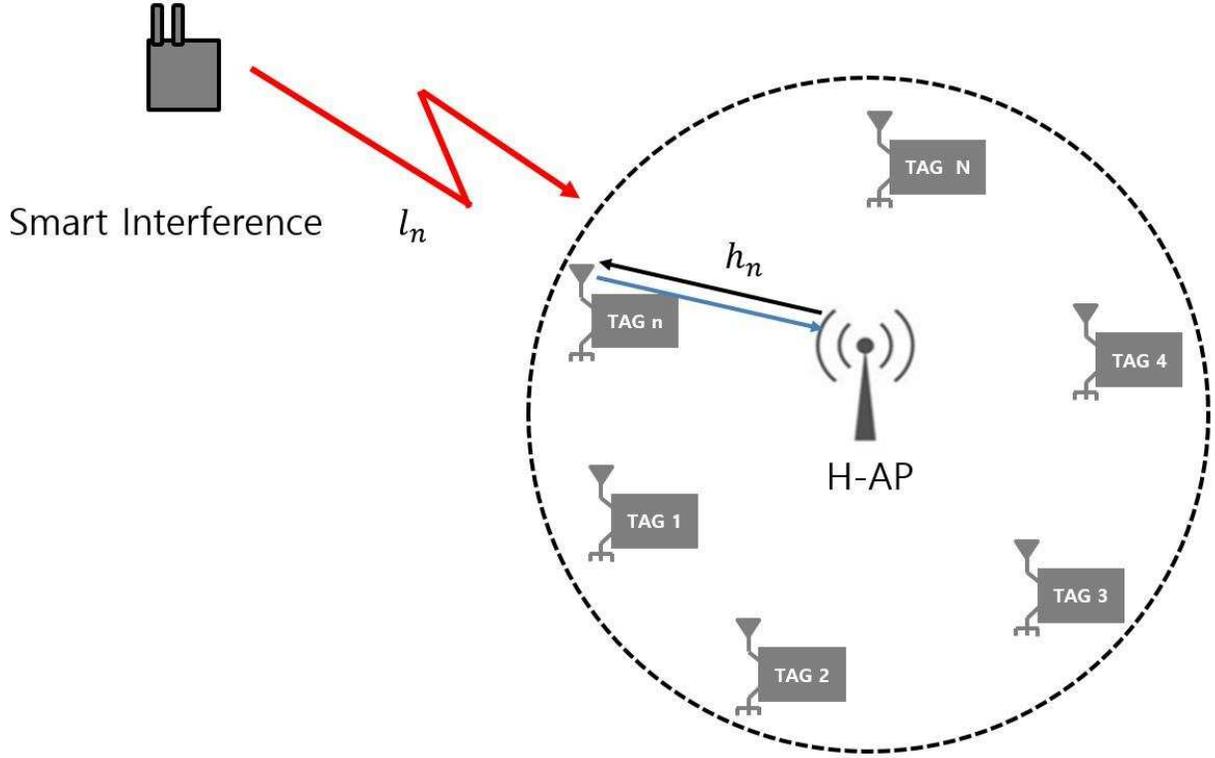}
\caption{The backscatter WSN under the smart interferer.}
\label{fig:model}
\end{figure}
Here, $G_t$ is the antenna gain of the H-AP and $G_r$ the antenna gain of the tag. $\lambda_B$ is the wavelength of the signal transmitted from the H-AP. $r_n$ is the distance between the H-AP and the $n$-th tag. After energy harvesting during the time $(1-\rho)T$, the tag performs backscattering to transmit its own data using the harvested energy during the time $\rho T$. The average transmit power for the $n$-th tag can be represented as:
\begin{equation}\label{eq:P_{B_,n}}
\begin{split}
P_{B,n} = E_n/(\rho T*t_n)
\end{split},
\end{equation}
where $t_n$ is a time slot equally allocated among users for TDMA systems. Hence, we can obtain the received SINR for the $n$-th tag at the receiver as \cite{AD_BC}:
\begin{equation}\label{eq:SINR_{k_,n}}
\begin{split}
SINR_{k_,n} = \frac{\delta_{k,n}P_{B,n}^{re}}{P_Il_n+N_B}=\frac{\delta_{k,n} P_{B_,n}h_n\mid\Gamma_0-\Gamma_1\mid^2}{P_Il_n+N_B}
\end{split},
\end{equation}
where $P^{re}_{B_,n}$ is the transmit power of the $n$-th backscatter tag, $\Gamma_0$ and $\Gamma_1$ are the reflection coefficients \cite{AD_BC}.
$N_B$ is noise power, which is assumed that the noise is negligible compared to the interference caused by the smart interferer. $\delta_{k,n}\in\{1,0\} $ the binary indicator allocating the $k$-th kub-channel to the H-AP, and $P_I$ the interference power at the interferer. $l_n $ is the channel gain from the interferer to the tags, which is defined as:
\begin{equation}\label{eq:l_n}
\begin{split}
l_n=G_iG_t\lambda^{2}_{j}/(4\pi r_j)^2 
\end{split}.
\end{equation}
Here, $G_i$ is the antenna gain of the smart interferer, $\lambda_j$ is the wavelength of the signal transmitted from the smart interferer, and $r_{j,n}$is the distance from the smart interferer to the $n$-th tag.
\subsection{Optimization Problem Formulation Based on Stackelberg Game Model}
In our game model, the smart interferer as a follower observes the transmit power of the backscatter WSN and maximizes its utility, then the backscatter WSN as a leader observes the interference power of the follower and maximizes its utility. Based on the game model, we formulated Stackelberg game as $\mathcal{G}=(\mathcal{N,A,U})$, where $\mathcal{N}$ is the set of players comprised of the smart interferer and the backscatter WSN. $\mathcal{A}$ is the constraint set, and $\mathcal{U}$ the utility set. In P1, the utility of the smart interferer, $\mathcal{U}_I$, is defined as how much smart interferer can reduce the SINR of the backscatter WSN under the smart interferer’s battery limit. As a mobile device, the smart interferer cannot enhance the interference effect by increasing the interference power indefinitely, and should interfere considering its limited battery power. To solve P1, we have assumed that the smart interferer can instantaneously observe and determine the transmit power of the backscatter WSN. $P^{re}_{B,n}$, and sub-channel in which backscatter signal is present, $ \delta_{k_,n}$. Further, we assumed that other parameter are known. In P2, $\mathcal{U}_B$ is the utility of the backscatter WSN to maximize the SINR with the least transmit power and harvest enough power for the backscatter WSN simultaneously. The optimization problem for the two players is formulated as:
\begin{equation}\label{eq:U_I}
\begin{split}
&(\mathsf{P1})~\max_{\substack{P_I \in \mathcal{A}_I}}\mathcal{U}_I(P_I \mid P_{B,n}^{re})\\
&=-\sum_{n=1}^N\frac{\delta_{k,n}\eta((1-\rho)/(\rho t_n))h^{2}_{n}P_{t,n}\mid\Gamma_0-\Gamma_1\mid^2}{P_I l_n+N_B}-C_IP_I,
\end{split}
\end{equation}
\begin{equation}\label{eq:s1}
\begin{split}
&\text{subject to:}~~~~~~\mathcal{A}_I=\{(P_I \mid P_{B,n}^{re})\mid P_I<P_{I,max}\}.~~~~~~~~
\end{split}
\end{equation}
\begin{equation}\label{eq:U_B}
\begin{split}
&(\mathsf{P2})~\max_{\substack{P_{B,n},\rho \in \mathcal{A}_B}}\mathcal{U}_B(P_{t,n},\rho \mid P_{I})\\
&=\mathrm{\sum_{k=1}^K\sum_{n=1}^N}\frac{\delta_{k,n}\eta((1-\rho)/(\rho t_n))h^{2}_{n}P_{t,n}\mid\Gamma_0-\Gamma_1\mid^2}{P_I l_n+N_B}-C_BP_{t,n},
\end{split}
\end{equation}
\begin{equation}\label{eq:s2}
\begin{split}
&\text{subject to:}\mathcal{A}_B=\left\{ (P_{t,n},\rho \mid P_I)\mid \begin{array}{c} 
SINR_{k,n}\geq SINR_{TH}\\
E_n \geq \rho TP_{B,TH}t_n\\
P_{t,n}<P_{t,max}\\
\delta_{k,n}=\{0,1\}\\
\rho = [0,1]\\
E_n>\eta(1-\rho)TP^{EH}_{TH}
\end{array} \right\}
\end{split}.
\end{equation}
Here, $P_{I,max}$ is the maximum interference power of the smart interferer, $P_{t,max}$ the maximum transmit power of the H-AP, and $P_{B,TH}$ the minimum input power of the tag for backscattering \cite{rf_har}. $P^{EH}_{TH}$ is a minimum harvested power for the backscatter WSN \cite{intro_rf2} and $SINR_{TH}$ is a minimum SINR to overcome the smart interferenIce \cite{BC_Shift} and enable reliable backscatter communications. $C_B$ and $C_I$ (unit: per unit power) are used to express the prices (or costs) of unit transmit power and unify the unit of the utility functions \cite{Stackelberg_Jamming1,Stackelberg_Jamming2}.

\section{Utility Optimization Based on Stackelberg Game}
In this section, we propose two utility optimization algorithms for the smart interferer and backscatter WSN to drive the Stackelberg equilibrium of the formulated game $\mathcal{G}$.
\subsection{Smart Interferer’s Best Response Strategy}
To maximally interrupt the backscatter communication of the backscatter WSN under limited energy constraint, we maximize the utility of the smart interferer, as a part of the Stackelberg game along with Largrange dual method.
\newtheorem{theorem}{Theorem}
\begin{theorem}
There exists a unique solution set $\{P_I\}$ to maximize $\mathcal{U}_I (P_I\mid P_{B,n}^{re})$.
\end{theorem}
\newtheorem{proof}{Proof}
\begin{proof}
If the Hessian matrix $\mathbf{H}(\mathcal{U}_I)$ is $negative~definite$, P1 is strictly concave. The Hessian matrix $\mathbf{H}(\mathcal{U}_I)$ is represented:
\end{proof}
\begin{equation}\label{eq:H1}
\begin{split}
\nonumber&\mathbf{H}(\mathcal{U}_I)=[\partial^2\mathcal{U}_I(P_I \mid P^{re}_{B,n})/\partial P^2_I].
\end{split}
\end{equation}
The second order partial derivative of $\mathcal{U}_I(P_I \mid P^{re}_{B,n})$ is:
\begin{equation}\label{eq:P_H1}
\begin{split}
&\frac{\partial^2\mathcal{U}_I(P_I \mid P^{re}_{B,n})}{\partial P^2_I}=-\sum_n^N\frac{2((1-\rho)/\rho)P_{t,n}A_{k,n}l^2_n\rho^2}{(P_I,l_n+N_B)^3},
\end{split}
\end{equation}
where $A_{k,n}=\delta_{k,n}\eta h_n^2\mid\Gamma_0-\Gamma_1\mid^2/t_n$. Therefore,  $\mathbf{H}(\mathcal{U}_I)$ is strictly concave and P1 has a critical point. Then, we can solve P1 by using the Lagrange multiplier method. We can write the Lagrange dual function with the Lagrange multiplier $\zeta_n$ as:
\begin{equation}\label{eq:L_1}
\begin{split}
&L_I(P_I,\zeta_n)=\mathcal{U}_I(P_I \mid P^{re}_{B,n})+\sum_{n=1}^N\zeta_n(P_{I,max}-P_I),
\end{split}
\end{equation}
The Lagrange dual problem of P1 can be formulated as:
\begin{equation}\label{eq:min_L_1}
\begin{split}
&\min_{\substack{\{\zeta_n\}}}\max_{\substack{\{P_I\}}} L_I(P_I,\zeta_n).
\end{split}
\end{equation}
Then, the optimal interference power $P_I^*$ is calculated (details is given in Appendix):
\begin{equation}\label{eq:op_PI}
\begin{split}
&P^*_I=\left[(\sqrt{l_n((1-\rho)/\rho)P_{t,n}A_{k,n}/(C_I+\zeta_n)}-N_B)/l_n\right]^+,
\end{split}
\end{equation}
where $[Z]^+=max\{ Z,0 \} $. As with multilevel water filling, the optimal values are carried out in (13). Moreover, the Lagrange multiplier $\zeta_n$ can be updated by using gradient methods in a distributed manner as:
\begin{equation}\label{eq:LI_P}
\begin{split}
&\zeta_n(T_I+1)=[\zeta_n(T_I)-\omega_1*(P_{I,max}-P^*_I)]^+,
\end{split}
\end{equation}
where the parameter $T_I$ is the number of iterations. The iteration step $\omega_1$ is a positive value like a learning rate to converge the algorithm faster. The proposed maximization algorithm for the smart interferer’s utility is given in Table I, lines 1-4.
\subsection{Backscatter WSN's Optimal Strategy}
In \cite{BC_Shift}, the backscatter WSN can shift a frequency band. We assume that the backscatter WSN conducts frequency shifting to a different sub-channel with less interference when the current sub-channel is under attack from the smart interferer. The binary indicator allocating the $k$-th sub-channel is defined:
\begin{equation}\label{eq:O_D}
\begin{split}
\delta^*_{k,n}=\begin{cases}
1,~\text{if}~P_{t,n}>P_{t,max}~\text{in the $k_{th}$ sub channel}
\\~~~~\text{where } k=\arg \max_{\substack{k}} SINR_{k,n}.\\
0,~otherwise.
\end{cases}
\end{split}
\end{equation}
Then, the backscatter WSN performs the power control on the shifted sub-channel against the smart interferer. Note that harvesting the interference power is neglected by the proposed sub-channel shifting strategy to avoid the interference. The P2 function can be formulated as an optimization problem:
\begin{theorem}
There exists a unique solution set $\{P_{t,n},\rho\}$ to maximize $\mathcal{U}_B (P_{t,n},\rho\mid P_I)$.
\end{theorem}
\begin{proof}
If the Hessian matrix $\mathbf{H}(\mathcal{U}_B)$ is $negative~definite$, P2 is strictly concave. The Hessian matrix $\mathbf{H}(\mathcal{U}_B)$ is represented as
\end{proof}
\begin{equation}\label{eq:H2}
\begin{split}
\nonumber&\mathbf{H}(\mathcal{U}_B)=\left[\begin{array}{cc}
    \partial^2\mathcal{U}_B/\partial P^2_{t,n} & \partial^2\mathcal{U}_B/\partial P_{t,n}\partial\rho\\
     \partial^2\mathcal{U}_B/\partial P_{t,n}\partial\rho & \partial^2\mathcal{U}_B/\partial \rho^2
\end{array}\right].
\end{split}
\end{equation}
The second-order partial derivative of $\mathcal{U}_B (P_{t,n},\rho\mid P_I)$ is:
\begin{table}[t]
\centering
\begin{tabular}{l c}
\hline
~~~~~~~~~~~~~~~~~~~~~~~~~~~~~~~~~~~~~~$\mathbf{TABLE I}$ \\
MAXIMIZATION ALGORITHM FOR SMART INTERFERER’S UTILITY\\
~~~~~~~~~~~~(LINES 1-4,  PERFORMED BY SMART INTERFERER)\\
1: \textbf{Input}:~$SINR_{TH}$,$P_{B,TH}$,$\eta$,$P_{B,max}$,$P_{I,max}$,$P_{t,n}$,$\rho$,$t_n$,$\nu$,$\zeta_n$,$\alpha_{k,n}$,$\beta_n$,\\
~~~$\mu_n$,$\gamma_{k,n}$,$\omega_1,\cdots,\omega_7>0$,$T$=$T_I$=$T_B$=1.\\
2: \textbf{Compute} the optimal $P^*_I$,$\zeta^*_n$ according to (13) and (14).\\
3: \textbf{If} $\mathcal{U}_I(T)>\mathcal{U}_I(T-1)$, \textbf{Retrun} $P^{T+1}_I$=$P^*_I$,$\zeta^{T+1}_n$=$\zeta^*_n$ and obtain the \\
~~~optimal utility set  $\mathcal{U}_I(T)$.\\
4: \textbf{else} go to line 2 and let $T_I$=$T_I$+1.\\
MAXIMIZATION ALGORITHM FOR BACKSCATTER WSN’S  UTILITY\\
~~~~~~~~~~~~(LINES 5-8,  PERFORMED BY BACKSCATTER WSN)\\
5: \textbf{Input}:~$P^*_I$.~\textbf{IF}$P_{t,n}>P_{t,max}$, \textbf{Shift} to a different sub-channel with 
\\~~~less interference according to (15).\\
6: \textbf{If} $\mathcal{U}_B(T)>\mathcal{U}_B(T-1)$, \textbf{Return}$P^{T+1}_{t,n}$=$P^*_{t,n}$,$\rho^{T+1}$=$\rho^*$,\\
~~~$\alpha^{T+1}_{k,n}$=$\alpha^*_{k,n}$,$\beta^{T+1}_n$=$\beta^*_n$,$\mu^{T+1}_n$=$\mu^*_n$,$\nu^{T+1}$=$\nu^*$,$\tau^{T+1}$=$\tau^*$,$\gamma^{T+1}_{k,n}$=$\gamma^*_{k,n}$\\
~~~and obtain the optimal utility set $\mathcal{U}_B(T)$ and let $T$=$T$+1.\\
7: \textbf{else} Update $P^{*}_{t,n}$=$P_{t,n}(T_B+1)$,$\rho^{*}$=$\rho(T_B+1)$,$\alpha^{*}_{k,n}$=$\alpha_{k,n}(T_B+1)$,\\
~~~$\beta^{*}_{n}$=$\beta_{n}(T_B+1)$,$\mu^{*}_{n}$=$\mu_{n}(T_B+1)$,$\nu^{*}$=$\nu(T_B+1)$,$\tau^{*}$=$\tau(T_B+1)$,\\
~~~$\gamma^{*}_{k,n}$=$\gamma_{k,n}(T_B+1)$~according to (21-28).\\
8: Go to line 6 and let $T_B$=$T_B$+1.\\
\hline
\end{tabular}
\end{table}
\begin{equation}\label{eq:P_U2}
\begin{split}
&\frac{\partial^2\mathcal{U}_B}{\partial P^2_{t,n}}=-\frac{(1-\rho )A_{k,n}}{4}\sqrt{\frac{C_I+\zeta_n}{l_n\rho(1-\rho)A_{k,n}P^3_{t,n}}}
\end{split},
\end{equation}
\begin{equation}\label{eq:P_U3}
\begin{split}
&\frac{\partial^2\mathcal{U}_B}{\partial \rho^2}=\frac{B_{k,n}(4\rho-3)}{4(\rho^4-\rho^3)(A_{k,n}P_{t,n}(\rho-\rho^2)/C_n)^{1/2}}
\end{split},
\end{equation}
\begin{equation}\label{eq:P_U4}
\begin{split}
&\frac{\partial^2\mathcal{U}_B}{\partial P_{t,n}\rho}=\frac{\partial^2\mathcal{U}_B}{\partial \rho P_{t,n}}=-\frac{D_{k,n}(1-2\rho)A_{k,n}(\rho-1)}{4C_n \sqrt{D_{k,n}(\rho-\rho^2)/C_n}}
\end{split},
\end{equation}
\begin{figure}[h]
\centering
\includegraphics[width=4.5in]{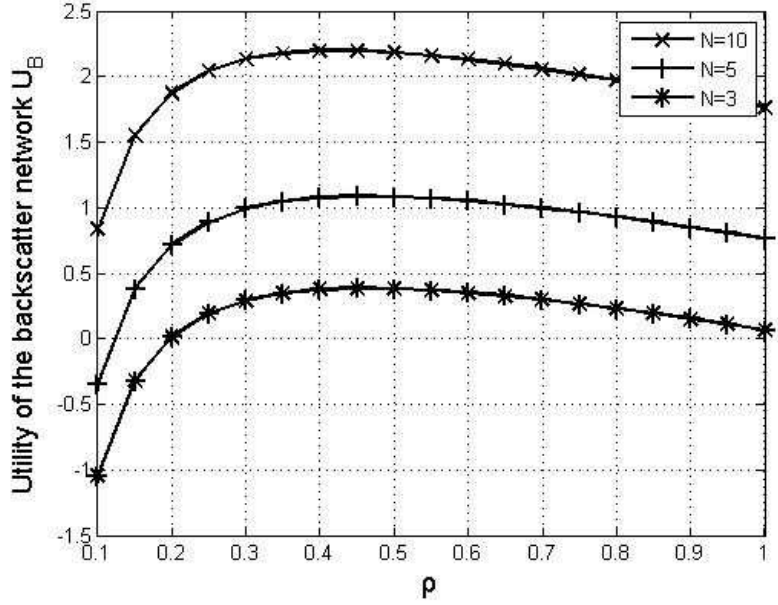}
\caption{Time-switching ratio $\rho$ for an optimum utility of the backscatter WSN}
\label{fig:model}
\end{figure}
\begin{figure}[h]
\centering
\includegraphics[width=4.5in]{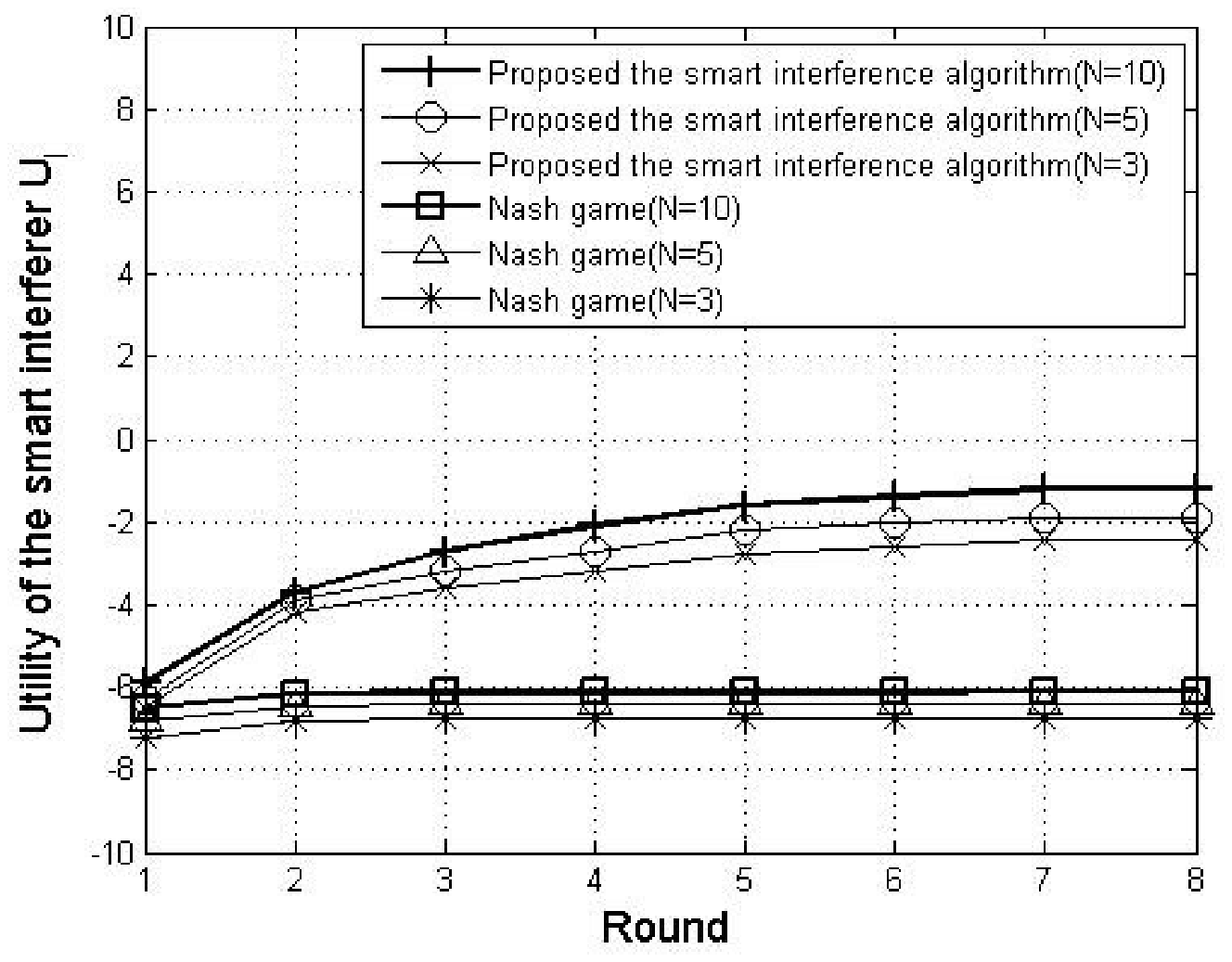}
\caption{Comparison of the smart interferer’s utility between the proposed algorithm and the one based on the Nash game}
\label{fig:model}
\end{figure}
\begin{figure}[h]
\centering
\includegraphics[width=4.5in]{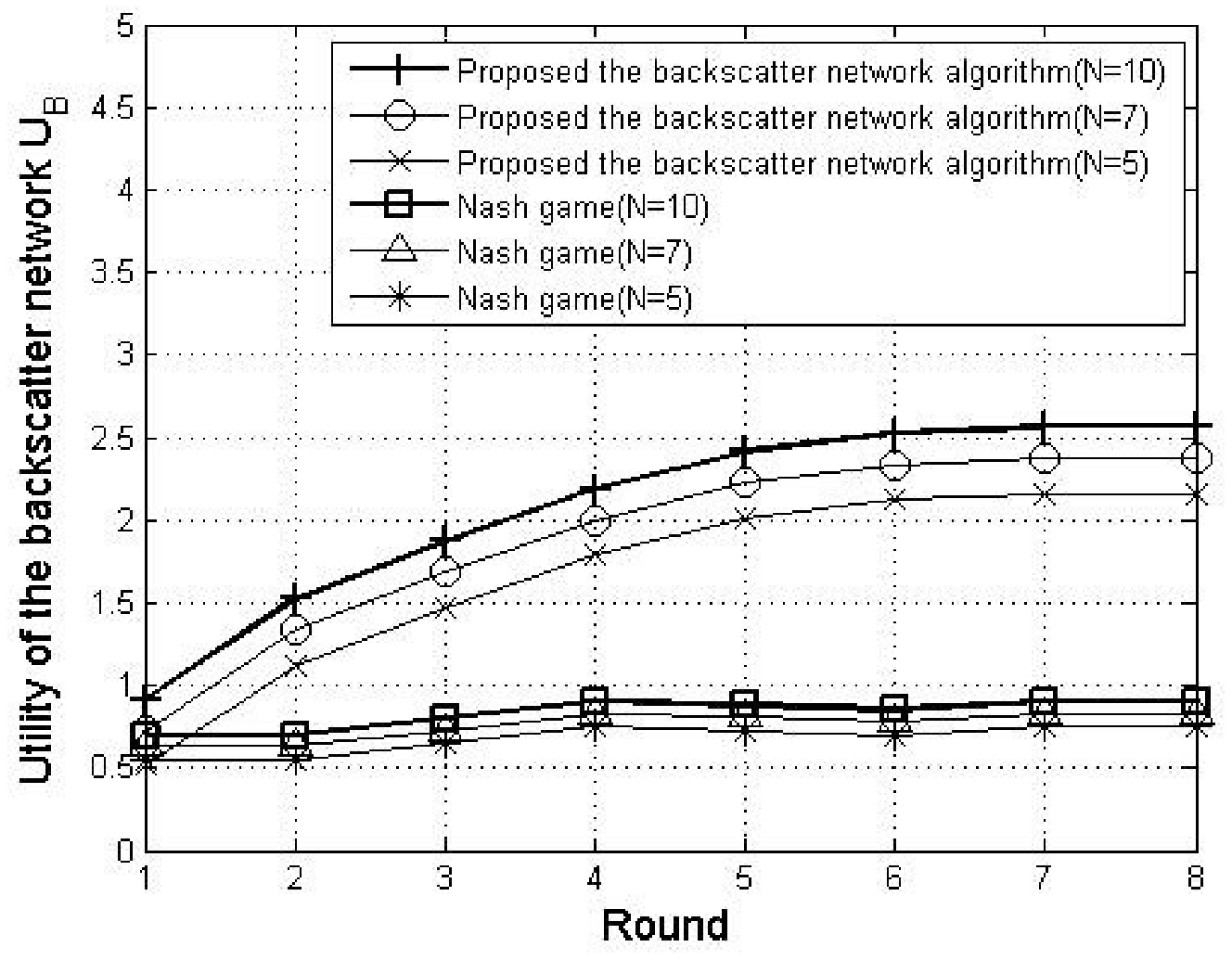}
\caption{Comparison of the backscatter WSN’s utility between the proposed algorithm and the one based on the Nash game}
\label{fig:model}
\end{figure}
where $B_{k,n}=l_n P_{t,n} A_{k,n}$ and $C_n=C_I+\zeta_n$. Therefore, $\mathbf{H}(\mathcal{U}_B)$ is strictly concave and P2 has a critical point. Then, we can solve P2 by using the Lagrange multiplier method. We write the dual Lagrange function of P2 as:
\begin{equation}\label{eq:L_2}
\begin{split}
&L_B(P_{t,n},\rho,\alpha_{k,n},\beta_n,\tau,\mu_n,\nu,\gamma_{k,n})=\mathcal{U}_B(P_{t,n},\rho \mid P_I)+\\
&\sum_{k=1}^K\sum_{n=1}^N \{\alpha_{k,n}(SINR_{k,n}-SINR_{TH})+\gamma_{k,n}(\eta h_nP_{t,n}\\
&-\eta T P^{EH}_{TH})\}+\sum_{n=1}^N\{\beta_n(P_{t,max}-P_{t,n})+\nu(1-\rho)\\
&\tau\rho+\mu_n(E_n-t_n\rho TP_{B,TH})\},
\end{split}
\end{equation}
where $\mu_n$ and $\nu$ are Lagrange multipliers. Therefore, we describe the dual Lagrange problem of P2 as:
\begin{equation}\label{eq:min_L_2}
\begin{split}
\min_{\substack{\{\ \alpha_{k,n},\beta_n,\tau,\mu_n,\nu,\gamma_{k,n}\}}}\max_{\substack{\{P_{t,n},\rho\}}} L_{B,n}(P_{t,n},\rho,\alpha_{k,n},\beta_n,\tau,\mu_n,\nu,\gamma_{k,n}).
\end{split}
\end{equation}
The optimal values are derived as (details are in Appendix):
\begin{equation}\label{eq:op_Pt}
\begin{split}
&P^*_{t,n}=\left[\frac{C_n}{A_{k,n}l_n(1-\rho)}(\frac{(1+\alpha_{k,n})(1-\rho)A_{k,n}}{2E_{k,n}})^2\right]^+,
\end{split}
\end{equation}
\begin{equation}\label{eq:op_rho}
\begin{split}
&\rho^*=\left[-1/4+1/2(G^{1/2}_{k,n}-(1/2-G_{k,n}-G^{-1/4}_{k,n})\right]^+,
\end{split}
\end{equation}
where $D_{k,n}=\eta(\rho-1)h_n(\mu_nT+\gamma_{k,n})-C_B-\beta_n$, $E_{k,n}=\eta h_nTP_{t,n}+t_nTP_{B,TH}+\nu-\tau, ~F_{k,n}=\frac{C_n}{A_{k,n}P_{t,n}}(\frac{(1+\alpha_{k,n})B_{k,n}}{2E_{k,n}})^2 ~\text{and}~ G_{k,n}=1/4-(\sqrt[3]{2}\cdot12F_{k,n}/(F_{k,n}(3\sqrt{768F_{k,n}+81}-27))^{1/3})+(-27F_{k,n}+3F_{k,n}\sqrt{768F_{k,n}+81})^{1/3}/3\sqrt[3]{2}.$ The Lagrange multipliers $\alpha_{k,n},\beta_n,\tau,\mu_n,\nu,\text{ and }\gamma_{k,n}$ are updated by using the gradient methods. The parameter $T_B$ the number of iterations. Constant coefficients $\omega_2,\cdots,\omega_7$ positive iteration steps as:
\begin{equation}\label{eq:LB_P1}
\begin{split}
&\alpha_{k,n}(T_B+1)=[\alpha_{k,n}(T_B)-\omega_2(SINR^*_{k,n}-SINR_{TH})]^+,
\end{split}
\end{equation}
\begin{equation}\label{eq:LB_P2}
\begin{split}
&\beta_{n}(T_B+1)=[\beta_{n}(T_B)-\omega_3(P_{t,max}-P^*_{t,n})]^+,
\end{split}
\end{equation}
\begin{equation}\label{eq:LB_P3}
\begin{split}
&\mu_{n}(T_B+1)=[\mu_{n}(T_B)-\omega_4(E^*_{n}-\rho TP_{B,TH}t_{n})]^+,
\end{split}
\end{equation}
\begin{equation}\label{eq:LB_P4}
\begin{split}
&\tau(T_B+1)=[\tau(T_B)-\omega_5\rho^*]^+,
\end{split}
\end{equation}
\begin{equation}\label{eq:LB_P5}
\begin{split}
&\nu(T_B+1)=[\nu(T_B)-\omega_6(1-\rho^*)]^+,
\end{split}
\end{equation}
\begin{equation}\label{eq:LB_P6}
\begin{split}
&\gamma_{k,n}(T_B+1)=[\gamma_{k,n}(T_B)-\omega_7(\eta h_nP_{t,n}-\eta TP^{EH}_{TH})]^+.
\end{split}
\end{equation}
The proposed maximization algorithm for backscatter WSN’s utility is given in Table I, lines 5-8, and performed at the backscatter WSN.

\section{Results}
For the numerical results, we set parameters as the total number of tags $N$=3, 5 and 10, $P_{B,TH}$=-18dBm \cite{rf_har}, $P^{EH}_{TH}$=-22dBm \cite{intro_rf2}, $SINR_{TH}$=10dB \cite{BC_Shift}, $C_I$=$C_B$=1 [8-9], $\eta$=0.5, $t_{t,n}$=1/$N$, $\Gamma_0$=1, $\Gamma_1$=-1 \cite{AD_BC}, $P_{B,max}$=20dBm \cite{Stackelberg_Jamming1}, $P_{I,max}$=30dBm \cite{Stackelberg_Jamming1}, $G_t$=$G_i$=6dBi, $G_r$=1.8dBi \cite{rf_har}, and $K$=14. The backscatter WSN operates in 2.4 GHz ISM band. We assume a scenario where the backscatter WSN has a coverage up to 5 meters in one office space in a building, and the interferer should be located far enough from the coverage of the backscatter WSN considering that it is physically easy to be detected and removed. On the other hand, to prevent the effect of the interference attack from becoming very small due to too far distance from the backscatter WSN, $r_{j,n}$is set to a proper distance of 10 meters. Based on the mobile-type smart jammer \cite{web} which could act as the smart interferer, the smart interferer is simulated using MATLAB. Fig. 2 shows the time-switching ratio ρ for the optimum utility of the backscatter WSN. We allocate the time-switching $\rho$ for RF energy harvesting using the Lagrange multiplier method. The utility of the backscatter WSN attains the maximum when the time-switching ratio $\rho$ is 0.45. Thus, we can allocate the optimal time-switching ratio for backscattering and RF energy harvesting in the backscatter WSN.

Fig. 3 shows the comparison of the results of the proposed algorithm of the smart interferer with those of the algorithm based on the Nash game. The Nash game does not consider the hierarchical structure, due to which the proposed algorithm of the smart interferer provides an increased utility over that of the algorithm based on the Nash game, which implies the proposed algorithm efficiently interferes with the backscatter WSN.

Fig. 4 demonstrates each game model reaches its own equilibrium where the utility no longer increases and converges, but the Stackelberg game-based curves converge to the higher utility than that of the Nash game-based curves at the last round 8. This means that applying the hierarchical structure-based Stackelberg game theory is a more realistic and appropriate way to solve the smart interference problem. It should be noted that the smart interferer in this paper is assumed to be the smartest one with malicious intention and to perform the proposed interferer’s utility maximization algorithm. If the smart interferer uses another algorithm which would probably be less efficient and uses higher power to interfere, the backscatter WSN's utility could be even higher due to faster energy depletion of the interferer.

\section{Conclusion}
In this paper, we have formulated the smart interference problem in the interaction between the backscatter WSN and the smart interferer. We have proposed two resource allocation algorithms based on the Stackelberg game and proved the existence of the Stackelberg equilibrium. The results showed that the proposed algorithms yield an increased utility over that of the algorithm based on the Nash game.
\appendix
\begin{center}
DERIVATION OF THE OPTIMAL VALUE (13)
\end{center}
By solving the following formula derived by the Karush-Kuhn-Tucker (KKT) conditions, we can simply derive the optimal values of $P_I^*$ :
$\partial L_I(P_I,\zeta_n)/\partial P_I$
\begin{equation}\label{eq:A_L1}
\begin{split}
=\frac{l_nP_{t,n}A_{k,n}(1-\rho)}{\rho(P_Il_n+N_B)^2}-C_I-\zeta_n
\begin{cases}
=0,~P_{I,n}>0,\\
<0,~\text{otherwise.}
\end{cases}
\end{split}
\end{equation}
\begin{center}
DERIVATION OF THE OPTIMAL VALUE (21-22)
\end{center}
The optimal values,$P_{t,n}^*$ and $\rho^*$, are derived by the KKT conditions, respectively, by solving the following equations:
$\partial L_B(P_{t,n},\rho,\alpha_{k,n},\beta_n,\tau,\mu_n,\nu,\gamma_{k,n})/\partial P_{t,n}$
\begin{equation}\label{eq:A_L2}
\begin{split}
=\frac{(1+\alpha_{k,n})(1-\rho)A_{k,n}}{2\sqrt{(1-\rho)B_{k,n}/C_n}}-D_{k,n}
\begin{cases}
=0,~P_{t,n}>0,\\
<0,~\text{otherwise.}
\end{cases}
\end{split}
\end{equation}
$\partial L_B(P_{t,n},\rho,\alpha_{k,n},\beta_n,\tau,\mu_n,\nu,\gamma_{k,n})/\partial\rho$
\begin{equation}\label{eq:A_L3}
\begin{split}
=-\frac{(1+\alpha_{k,n})B_{k,n}}{2\rho^2\sqrt{(1-\rho)P_{t,n}A_{k,n}/\rho C_n}}-E_{k,n}
\begin{cases}
=0,~\rho>0,\\
<0,~\text{otherwise.}
\end{cases}
\end{split}
\end{equation}

\end{document}